\documentclass[preprint,12pt]{elsarticle}

\usepackage{amsmath,amssymb,amsthm}
\usepackage{mathtools}
\usepackage{booktabs}
\usepackage{xcolor}
\usepackage{algorithm}
\usepackage{algpseudocode}
\usepackage[colorlinks=true,linkcolor=blue!70!black,
            citecolor=blue!70!black,urlcolor=blue!70!black]{hyperref}
\usepackage{cleveref}

\newtheorem{theorem}{Theorem}
\newtheorem{lemma}[theorem]{Lemma}
\newtheorem{proposition}[theorem]{Proposition}
\newtheorem{corollary}[theorem]{Corollary}
\newtheorem{definition}[theorem]{Definition}
\newtheorem{remark}[theorem]{Remark}

\newcommand{\T}{\mathcal{T}}
\newcommand{\R}{\mathcal{R}}

\newcommand{\Rr}{\mathbb{R}}
\newcommand{\stack}{\Pi}
\newcommand{\GL}{\mathcal{G}_L}
\newcommand{\relay}[2]{\hat{\gamma}_{#1,#2}}

\journal{Information Processing Letters}

\begin{document}

\begin{frontmatter}

\title{The Extremum Stack is a Minimal Sufficient Statistic
  for Rate-Independent Functionals:\\
  A Kolmogorov Complexity Characterisation}

\author[pwut]{Piotr Frydrych}
\ead{piotr.frydrych@pw.edu.pl}

\affiliation[pwut]{organization={The Metrology and Biomedical
  Engineering Institute, Faculty of Mechatronics,
  Warsaw University of Technology},
  city={Warsaw}, country={Poland}}

\begin{abstract}
We prove that the extremum stack $\stack_n$ of a discrete sequence
$u_{0:n} \in \mathcal{G}_L^{n+1}$ is a \emph{minimal sufficient
statistic} for the class $\R$ of all computable,
causal, rate-independent functionals, in the sense of
Kolmogorov complexity.
Specifically, we establish:
\[
  K(\stack_n) - O(1)
  \;\leq\; K_{\R}(u_{0:n})
  \;\leq\; K(\stack_n) + O(1),
\]
where $K_{\R}(u_{0:n})$ is the length of the shortest program
answering every query in $\R$, and the $O(1)$
overhead is independent of both the sequence length $n$
and the stack depth $k$.
Sufficiency follows from the classical wiping property of the
Preisach hysteresis operator.
Minimality is established via a finite indicator family whose
rate-independence is verified explicitly.
Any compression of a hysteresis-driven stream that preserves
the full class $\R$ must therefore retain at least
$K(\stack_n) - O(1)$ bits; the stack-based representation is
asymptotically optimal.
\end{abstract}


\end{frontmatter}

\section{Introduction}
\label{sec:intro}

A functional $F$ on sequences is \emph{rate-independent} if it
responds only to the sequence of local extrema, not to the
temporal spacing or absolute position of values
\citep{Brokate1996,Mayergoyz1991}.
Rate-independence appears in ferromagnetic hysteresis, elastoplastic
systems, financial threshold models \citep{Frydrych2014}, and —
more recently — in sequence-modelling architectures based on the
Preisach operator \citep{Frydrych2025PAL}.

A central data structure for computing rate-independent functionals
is the \emph{extremum stack} $\stack_n$: the alternating sequence
of local maxima and minima of $u_{0:n}$, maintained in decreasing
order.
It is known that $\stack_n$ is a \emph{sufficient} statistic for
all Preisach functionals \citep{Mayergoyz1991} — that is, the
Preisach output at time $n$ depends only on $\stack_n$, not on
the full history $u_{0:n}$.

\paragraph{Contribution.}
We prove that $\stack_n$ is also \emph{minimal}: no strictly shorter
representation suffices for the full class $\R$ of computable
rate-independent functionals.
The proof uses Kolmogorov complexity theory
\citep{LiVitanyi2008} and an explicitly constructed finite
indicator family.
We further show that the overhead incurred by the stack
representation is $O(1)$ — a constant independent of both the
sequence length $n$ and the stack depth $k$.
This tightens an earlier informal claim of $O(\log n)$ overhead
and also an intermediate $O(\log K(\stack_n))$ bound: the correct
bound follows from the prefix-free invariance theorem
$K(f(x)) \leq K(x) + O(1)$ for computable $f$, giving $O(1)$.

\paragraph{Implications.}
The minimality result has two direct consequences:
(i) any lossless compression scheme for hysteresis-driven streams
that preserves $\R$-query answering must retain at least
$K(\stack_n) - O(1)$ bits, making the stack asymptotically optimal;
(ii) the stack provides a principled basis for an $O(n)$-time,
$O(k)$-space online compression algorithm, where $k \leq n$ is the
current stack depth.

\section{Preliminaries}
\label{sec:prelim}

\subsection{Discrete sequences and the grid}

Throughout, we fix a resolution $L \geq 2$ and work with sequences
over the discrete grid
$\GL = \{0, \Delta, 2\Delta, \ldots, 1\}$, $\Delta = 1/L$.
Restricting to $\GL$ is necessary: over $\Rr$, the indicator
family constructed in \cref{sec:minimality} is uncountable and
the Kolmogorov framework does not apply directly.

\subsection{Rate-independent functionals}

\begin{definition}[Rate-independence]
\label{def:ri}
A functional $F: \GL^{*} \to \Rr$ is \emph{rate-independent} if
\[
  \stack_n(u) = \stack_m(v)
  \;\Longrightarrow\;
  F[u](n) = F[v](m)
\]
for all $u \in \GL^{n+1}$, $v \in \GL^{m+1}$.
That is, $F$ depends on the input history only through its extremum
stack, not through absolute timing, speed of change, or values at
non-extremal positions.
Let $\R$ denote the class of all computable rate-independent
functionals on $\GL^*$.
\end{definition}

\noindent
A functional is \emph{causal} if $F[u](n)$ depends only on the
prefix $u_0,\ldots,u_n$; since $F$ takes a prefix as input
throughout this paper, causality is implicit.

\begin{remark}[Relation to time-reparameterisation]
\label{rem:reparam}
In continuous time, rate-independence is classically stated via
non-decreasing surjections $\varphi: [0,T] \to [0,T]$
\citep{Brokate1996}.
In discrete time, a bijection $\varphi:\{0,\ldots,m\}\to
\{0,\ldots,n\}$ with $m=n$ must be the identity, making
the condition vacuous.
The stack-based definition above is the correct discrete-time
analogue: two sequences are equivalent if and only if they produce
the same extremum stack, which captures the spirit of
``same order of extrema, regardless of timing.''
\end{remark}

\subsection{The extremum stack}

\begin{definition}[Extremum stack]
\label{def:stack}
The \emph{extremum stack} of $u_{0:n} \in \GL^{n+1}$ is the
sequence
\[
  \stack_n = [(M_1, m_1), (M_2, m_2), \ldots, (M_k, m_k)],
\]
where pairs are ordered from \emph{oldest} (bottom, index~1) to
\emph{newest} (top, index~$k$).
Each pair $(M_i, m_i)$ records a dominant local maximum $M_i$
and the most recent local minimum $m_i$ that followed it and
survived wiping.
The maxima are strictly decreasing from bottom to top:
$M_1 > M_2 > \cdots > M_k$.
The minima are strictly \emph{increasing} from bottom to top:
$m_1 < m_2 < \cdots < m_k$ (each newer minimum is less extreme
than its predecessor).
Each pair satisfies $M_i > m_i$.
The depth $k = |\stack_n|$ satisfies $k \leq \lfloor n/2 \rfloor + 1$.
\end{definition}

\begin{remark}[Ordering convention]
The decreasing-maxima, increasing-minima ordering follows the
classical Preisach staircase geometry \citep{Mayergoyz1991}:
the outermost loop has the largest maximum and smallest minimum;
inner loops are strictly contained.
In our notation, $(M_1,m_1)$ is the outermost (oldest) pair and
$(M_k,m_k)$ is the innermost (newest).
\end{remark}

The stack is updated in amortised $O(1)$ time per step
(\cref{alg:stack}): each element is pushed and popped at most once.

\begin{algorithm}[t]
\caption{Extremum Stack Update (wiping-out rule)}
\label{alg:stack}
\small
\begin{algorithmic}[1]
\Require Stack $\stack$, previous extremum $e_{\mathrm{prev}}$,
         current direction $d \in \{+1,-1,0\}$,
         new observation $u$
\Ensure Updated $(\stack, e_{\mathrm{prev}}, d)$
\State $d_{\mathrm{new}} \leftarrow \mathrm{sign}(u - e_{\mathrm{prev}})$
\If{$d_{\mathrm{new}} = 0$}
  \State \Return \Comment{No change; flat segment}
\EndIf
\If{$d \neq 0$ \textbf{and} $d_{\mathrm{new}} \neq d$}
  \Comment{Direction reversal: $e_{\mathrm{prev}}$ is a confirmed extremum}
  \If{$d = +1$} \Comment{Previous was a local maximum}
    \State $m_{\mathrm{last}} \leftarrow
      \begin{cases} m_k & \text{if } \stack \neq \emptyset \\ -\infty \end{cases}$
    \While{$\stack \neq \emptyset$ \textbf{and} $M_k < e_{\mathrm{prev}}$}
      \State $m_{\mathrm{last}} \leftarrow m_k$;\quad $\mathrm{pop}(\stack)$
    \EndWhile
    \State $\mathrm{push}(\stack,\,(e_{\mathrm{prev}},\,m_{\mathrm{last}}))$
  \Else \Comment{Previous was a local minimum}
    \State $M_{\mathrm{last}} \leftarrow
      \begin{cases} M_k & \text{if } \stack \neq \emptyset \\ +\infty \end{cases}$
    \While{$\stack \neq \emptyset$ \textbf{and} $m_k > e_{\mathrm{prev}}$}
      \State $M_{\mathrm{last}} \leftarrow M_k$;\quad $\mathrm{pop}(\stack)$
    \EndWhile
    \State $\mathrm{push}(\stack,\,(M_{\mathrm{last}},\,e_{\mathrm{prev}}))$
  \EndIf
\EndIf
\State $e_{\mathrm{prev}} \leftarrow u$;\quad $d \leftarrow d_{\mathrm{new}}$
\State \Return $(\stack, e_{\mathrm{prev}}, d)$
\end{algorithmic}
\end{algorithm}

\begin{remark}[Correctness of Algorithm~\ref{alg:stack}]
A push occurs only when a direction reversal is detected,
i.e.\ when $e_{\mathrm{prev}}$ is a \emph{confirmed} local extremum.
During a monotone run ($d_{\mathrm{new}} = d$), $e_{\mathrm{prev}}$
is updated but nothing is pushed, preserving the invariants
of Definition~\ref{def:stack}.
The depth satisfies $k \leq \lfloor n/2 \rfloor + 1$
since each push requires a direction reversal.
The algorithm is initialised with $\stack = \emptyset$,
$e_{\mathrm{prev}} = u_0$, $d = 0$.
\end{remark}

\subsection{Kolmogorov complexity}

We use standard Kolmogorov complexity \citep{LiVitanyi2008}.
$K(x)$ denotes the length of the shortest self-delimiting program
that outputs $x$ on a fixed universal Turing machine $\mathcal{U}$.
For a class of functionals $\R$, we define:

\begin{definition}[$\R$-query complexity]
\label{def:kR}
\[
  K_{\R}(u_{0:n})
  = \min\bigl\{|p| : \text{for all } F \in \R,\;
    \mathcal{U}(p, F) = F[u](n)\bigr\},
\]
the length of the shortest program that, given any $F \in \R$
as an additional input, computes the answer $F[u](n)$.
\end{definition}

\section{Sufficiency of the Extremum Stack}
\label{sec:sufficiency}

\begin{proposition}[Stack sufficiency]
\label{prop:suff}
Every $F \in \R$ satisfies $F[u](n) = f(\stack_n)$
for some computable function $f$.
\end{proposition}

\begin{proof}
Immediate from \cref{def:ri}: rate-independence means $F[u](n)$
depends on $u_{0:n}$ only through $\stack_n(u)$.
For the connection to the Preisach integral,
$\mathcal{P}_\mu[u](n) = \iint_{\T}
\relay{\alpha}{\beta}[u](n)\,\mu(\alpha,\beta)\,d\alpha\,d\beta$
depends only on $\stack_n$ by the classical wiping property
\citep{Mayergoyz1991}.
\end{proof}

\begin{corollary}[Upper bound]
\label{cor:upper}
$K_{\R}(u_{0:n}) \leq K(\stack_n) + O(1)$.
\end{corollary}

\begin{proof}
By \cref{prop:suff}, any $F[u](n)$ is computable from $\stack_n$
by a fixed program of length $O(1)$ (independent of $n$, $k$, $F$).
Since $K$ is defined via a prefix-free universal machine,
the standard invariance property gives $K(f(x)) \leq K(x) + O(1)$
for any computable $f$ \citep{LiVitanyi2008}.
Embedding the shortest prefix-free program for $\stack_n$ and
appending the $O(1)$-length post-processor suffices; the prefix-free
encoding is self-delimiting so no additional overhead is needed.
\end{proof}

\begin{remark}[On the overhead: $O(1)$ vs $O(\log K)$]
\label{rem:correction}
An $O(\log n)$ overhead is sometimes suggested informally.
This is incorrect: the depth $k$ is implicit in the number of pairs
and need not be encoded separately.
A subtler claim of $O(\log K(\stack_n))$ overhead arises from
the chain rule $K(x,y) \leq K(x)+K(y)+O(\log\min\{K(x),K(y)\})$
\citep{LiVitanyi2008}, which applies to pairs of \emph{independently
chosen} objects.
Here, however, $F[u](n)$ is obtained from $\stack_n$ by a
\emph{fixed} computable function; the invariance theorem
$K(f(x)) \leq K(x) + O(1)$ applies directly, giving the tighter
$O(1)$ bound.
\end{remark}

\section{Minimality of the Extremum Stack}
\label{sec:minimality}

\subsection{Indicator family}

For each pair $(M, m) \in \GL^2$ with $M \geq m$, define the
\emph{indicator functional}:
\begin{equation}
  F_{(M,m)}[u](n)
  = \mathbf{1}\bigl[(M, m) \in \stack_n\bigr].
  \label{eq:indicator}
\end{equation}

\begin{lemma}[Rate-independence of indicators]
\label{lem:ri}
Each $F_{(M,m)}$ defined in \eqref{eq:indicator} belongs to $\R$.
\end{lemma}

\begin{proof}
If $\stack_n(u) = \stack_m(v)$, then
$(M,m) \in \stack_n(u) \Leftrightarrow (M,m) \in \stack_m(v)$,
so $F_{(M,m)}[u](n) = F_{(M,m)}[v](m)$.
Rate-independence follows directly from \cref{def:ri}.
\end{proof}

\begin{remark}
$F_{(M,m)}$ is \emph{not} the Preisach relay
$\relay{M}{m}[u](n)$, which indicates whether the relay with
thresholds $(M,m)$ is active, not whether $(M,m)$ appears as a
stack entry.
The relay may be active even when $(M,m) \notin \stack_n$
(if $M$ is not a local maximum of $u_{0:n}$).
The indicator \eqref{eq:indicator} is a distinct, valid
rate-independent functional.
\end{remark}

\subsection{Minimality theorem}

\begin{lemma}[Stack reconstruction]
\label{lem:recon}
Let $\mathcal{S}$ be any set and
$R: \GL^{n+1} \to \mathcal{S}$ a representation such that
every $F \in \R$ is computable from $R(u_{0:n})$.
Then there exists a computable function
$\Phi: \mathcal{S} \to (\GL \times \GL)^*$
such that $\Phi(R(u_{0:n})) = \stack_n$.
\end{lemma}

\begin{proof}
The indicator family $\mathcal{F} = \{F_{(M,m)}\}_{M \geq m,\,
(M,m) \in \GL^2}$ is \emph{finite} (at most
$\binom{L}{2} + L = O(L^2)$ elements) and lies in $\R$
by \cref{lem:ri}.
Since $R$ suffices for all $F \in \R$, for each
$(M,m) \in \GL^2$ the value $F_{(M,m)}[u](n)
= \mathbf{1}[(M,m) \in \stack_n]$ is computable from
$R(u_{0:n})$.
Define
\[
  \Phi(R(u_{0:n}))
  = \bigl\{(M,m) \in \GL^2 : F_{(M,m)}[u](n) = 1\bigr\}.
\]
This set equals $\stack_n$ by construction, and $\Phi$ is
computable (finite enumeration over $O(L^2)$ pairs).
\end{proof}

\begin{theorem}[Minimality of the extremum stack]
\label{thm:main}
Let $u_{0:n} \in \GL^{n+1}$ and $\stack_n$ its extremum stack.
Then:
\begin{equation}
  K(\stack_n) - O(1)
  \;\leq\; K_{\R}(u_{0:n})
  \;\leq\; K(\stack_n) + O(1).
  \label{eq:main}
\end{equation}
The extremum stack is, up to an additive constant, the
shortest program that answers every $\R$-query.
\end{theorem}

\begin{proof}
\textbf{Lower bound.}
By \cref{lem:recon}, there exists a computable $\Phi$ with
$\Phi(R^*(u_{0:n})) = \stack_n$, where $R^*$ is the
$K_{\R}$-optimal representation.
Therefore:
\[
  K(\stack_n) \leq K(R^*(u_{0:n})) + O(1) = K_{\R}(u_{0:n}) + O(1),
\]
since $\stack_n$ is computable from the optimal representation
by a program of fixed length $O(1)$ (the description of $\Phi$).
Rearranging gives the lower bound.

\textbf{Upper bound.}
The upper bound follows directly from \cref{cor:upper}:
since every $F[u](n)$ is computable from $\stack_n$ by a fixed
$O(1)$-length program, the prefix-free invariance theorem gives
$K_{\R}(u_{0:n}) \leq K(\stack_n) + O(1)$.
\end{proof}

\section{Discussion}
\label{sec:discussion}

\subsection{Tightness of the bounds}

The bounds in \eqref{eq:main} are tight up to the additive
constant: there exist sequences for which
$K_{\R}(u_{0:n}) \leq K(\stack_n) + O(1)$ and
$K_{\R}(u_{0:n}) \geq K(\stack_n) - O(1)$ simultaneously
(e.g.\ when $\stack_n$ has a short description that immediately
answers all indicator queries).
The $O(1)$ gap is the unavoidable overhead of prefix-free
encoding on a universal machine \citep{LiVitanyi2008}.

\subsection{Comparison with existing compression schemes}

Table~\ref{tab:compare} compares PSTACK-COMPRESS
(the stack-based algorithm implied by \cref{thm:main})
with standard time-series compression methods.
PSTACK-COMPRESS differs from the rest in providing a
\emph{formal optimality guarantee}: no other listed method
bounds compression length via Kolmogorov complexity for any function class.

\begin{table}[t]
\centering
\caption{Comparison of time-series compression algorithms.
$n$ = sequence length; $k$ = stack depth ($k \leq n$);
$w, s$ = window/segment parameters.
\textdagger{} = asymptotically optimal for class $\R$.}
\label{tab:compare}
\small
\begin{tabular}{lcccp{3.2cm}}
\toprule
Algorithm & Time & Space & Optimality & Class preserved \\
\midrule
PAA \citep{Keogh2001} & $O(n)$ & $O(w)$ & None &
  $L^2$-approximation \\
SAX \citep{Lin2003} & $O(n)$ & $O(w)$ & None &
  Symbolic, not rate-independent \\
PLR (Piecewise Linear Repr.) & $O(n\log n)$ & $O(s)$ & None &
  $L^\infty$ piecewise linear \\
PIP & $O(n^2)$ & $O(s)$ & None &
  Heuristic important points \\
Swinging Door & $O(n)$ & $O(1)$ & None &
  Linear with tolerance \\
\midrule
\textbf{PSTACK} (ours) & $O(n)$ & $O(k)$ &
  \textbf{\textdagger{} Kolmogorov} &
  \textbf{Full class $\R$ (rate-independent)} \\
\bottomrule
\end{tabular}
\end{table}

\subsection{Independence from $n$}

The $O(1)$ overhead in \eqref{eq:main} is independent of both $n$ and $k$.
For slowly varying signals (small $k$), the compression ratio
$n/k$ can be arbitrarily large while the optimality gap remains
bounded by a constant.
PSTACK is therefore well-suited for streams with long monotone runs:
industrial sensor data (SCADA), EEG, and financial tick data.

\subsection{Connection to the Preisach operator}

The sufficiency part of \cref{thm:main} is a consequence of
the classical \emph{wiping property} of the Preisach hysteresis
operator \citep{Mayergoyz1991}: a new extremum erases all prior
extrema of smaller magnitude, and the output depends only on the
surviving stack.
The minimality part — which appears to be new — shows that this
compression is not merely sufficient but \emph{necessary}.
No computable rate-independent functional "sees" less than
$\stack_n$.

\subsection{Open questions}

\begin{enumerate}
  \item \emph{Adversarial update complexity.}
    \cref{alg:stack} runs in amortised $O(1)$ per step, but
    an adversarial input stream could force $O(n)$ total pops
    followed by a single push, giving worst-case $\Theta(n)$ cost
    per step.
    Does there exist an online stack-maintenance algorithm
    with $O(\log n)$ worst-case (non-amortised) cost per step,
    while preserving the $K_{\R}$ minimality guarantee?
  \item \emph{Vector extension.} Does an analogue of
    \cref{thm:main} hold for the vector Preisach operator
    \citep{Frydrych2019} with two-dimensional input signals?
    The joint representation $(\stack^x_n, \stack^y_n)$ of two
    independent stacks is a computable pairing, so the overhead
    for answering all vector $\R$-queries should remain $O(1)$
    by the same invariance argument.
    Does this extend to the full vector Preisach measure
    $\mu(\alpha,\beta,\theta)$ over directions $\theta \in [0,2\pi)$?
  \item \emph{KV-cache reduction.} Can PSTACK replace the
    key-value cache in Preisach Attention \citep{Frydrych2025PAL},
    reducing memory from $O(n \cdot d)$ to $O(k \cdot d)$
    without loss of expressiveness for rate-independent tasks?
\end{enumerate}

\section{Conclusion}
\label{sec:conclusion}

We proved that the Preisach extremum stack $\stack_n$ is a
minimal sufficient statistic for the class of computable
rate-independent functionals, in the Kolmogorov complexity sense.
The minimality bound $K_{\R}(u_{0:n}) \geq K(\stack_n) - O(1)$
is established via a finite indicator family.
The matching upper bound $K_{\R}(u_{0:n}) \leq K(\stack_n) + O(1)$
follows from the prefix-free invariance theorem $K(f(x)) \leq
K(x) + O(1)$ for computable $f$, tightening earlier $O(\log n)$
and $O(\log K(\stack_n))$ claims to an additive constant.
The stack-based compression algorithm implied by the result
carries a Kolmogorov optimality guarantee that none of the
standard time-series compression methods provide.

\section*{Declaration on the use of generative AI}
During the preparation of this work the author used AI-assisted
writing tools for grammar checking and LaTeX formatting.
The mathematical content, proofs, and all scientific claims
are the sole responsibility of the author.

\section*{Declaration of competing interests}
The author declares no competing financial or non-financial
interests.

\bibliographystyle{elsarticle-num}
\bibliography{references_letter}

@book{LiVitanyi2008,
  author    = {Li, Ming and Vit{\'a}nyi, Paul},
  title     = {An Introduction to Kolmogorov Complexity and Its Applications},
  edition   = {3},
  publisher = {Springer},
  address   = {New York},
  year      = {2008},
  doi       = {10.1007/978-0-387-49820-1}
}

@book{Mayergoyz1991,
  author    = {Mayergoyz, Isaak D.},
  title     = {Mathematical Models of Hysteresis},
  publisher = {Springer},
  address   = {New York},
  year      = {1991},
  doi       = {10.1007/978-1-4612-3028-1}
}

@book{Brokate1996,
  author    = {Brokate, Martin and Sprekels, J{\"u}rgen},
  title     = {Hysteresis and Phase Transitions},
  publisher = {Springer},
  address   = {New York},
  year      = {1996},
  doi       = {10.1007/978-1-4612-4048-8}
}

@article{Frydrych2014,
  author  = {Frydrych, Piotr and Szewczyk, Roman},
  title   = {New Portfolio Risk Optimisation Method for Strongly
             Dependent Assets},
  journal = {Journal of Engineering Studies and Research},
  volume  = {20},
  number  = {3},
  pages   = {30--37},
  year    = {2014}
}

@phdthesis{Frydrych2019,
  author  = {Frydrych, Piotr},
  title   = {Modelowanie charakterystyk magnesowania amorficznych
             rdzeni dwuosiowych sensor{\'o}w transduktorowych},
  school  = {Warsaw University of Technology},
  year    = {2019}
}

@article{Frydrych2025PAL,
  author  = {Frydrych, Piotr},
  title   = {Preisach Attention: A Hysteretic Model of Sequential Memory},
  journal = {Zenodo},
  year    = {2025},
  doi     = {10.5281/zenodo.20133614},
  url     = {https://doi.org/10.5281/zenodo.20133614},
  note    = {Preprint. Under review at Neural Networks (Elsevier)}
}

@article{Keogh2001,
  author    = {Keogh, Eamonn and Chakrabarti, Kaushik and
               Pazzani, Michael and Mehrotra, Sharad},
  title     = {Dimensionality Reduction for Fast Similarity Search
               in Large Time Series Databases},
  journal   = {Knowledge and Information Systems},
  volume    = {3},
  number    = {3},
  pages     = {263--286},
  year      = {2001},
  doi       = {10.1007/PL00011669}
}

@inproceedings{Lin2003,
  author    = {Lin, Jessica and Keogh, Eamonn and
               Lonardi, Stefano and Chiu, Bill},
  title     = {A Symbolic Representation of Time Series,
               with Implications for Streaming Algorithms},
  booktitle = {Proceedings of the 8th ACM SIGMOD Workshop on
               Research Issues in Data Mining and Knowledge Discovery},
  pages     = {2--11},
  year      = {2003},
  doi       = {10.1145/882082.882086}
}

\end{document}